\documentclass[conference]{IEEEtran}

\usepackage{graphicx}
\usepackage{amsmath}
\usepackage{amssymb}
\usepackage{amsthm}

\usepackage{array}

\newtheorem{theorem}{Theorem}

\newtheorem{proposition}{Proposition}

\begin{document}

\title{Feedback Reduction for MIMO Broadcast Channel with Heterogeneous Fading}

\author{\IEEEauthorblockN{Jin-Hao Li and Hsuan-Jung Su}
\IEEEauthorblockA{Graduate Institute of Communication Engineering\\
                  Department of Electrical Engineering\\
                  National Taiwan University, Taipei, Taiwan\\
                  Email: jinghaw2003@gmail.com, hjsu@cc.ee.ntu.edu.tw}
}


\maketitle

\begin{abstract}
This paper considers feedback load reduction for multiuser multiple input multiple output (MIMO) broadcast channel where the users' channel distributions are not homogeneous. A cluster-based feedback scheme is proposed such that the range of possible signal-to-noise ratio (SNR) of the users are divided into several clusters according to the order statistics of the users' SNRs. Each cluster has a corresponding threshold, and the users compare their measured instantaneous SNRs with the thresholds to determine whether and how many bits they should use to feed back their instantaneous SNRs. If a user's instantaneous SNR is lower than a certain threshold, the user does not feed back. Feedback load reduction is thus achieved. For a given number of clusters, the sum rate loss using
the cluster-based feedback scheme is investigated. Then the minimum
number of clusters given a maximum tolerable sum rate
loss is derived. Through simulations, it is shown that, when the number of users is large, full multiuser diversity can be
achieved by the proposed feedback
scheme, which is more efficient than the conventional schemes.
\end{abstract}
%

\section{Introduction}
Multiuser diversity can significantly improve system throughput when
the users suffer channel fluctuations. The performance gain of
multiuser diversity grows with the number of users when the scheduler
performs maximum throughput scheduling
\cite{Pv_Tse02}. For the broadcast channel, the dirty paper coding (DPC) \cite{M_H_M_Costa1983}
achieves a sum rate which was shown in \cite{J_Cioffi_04} to have a
growth rate $M\log\log(K)$, where $M$ is the number of transmit
antennas at the base station (BS) and $K$ is the number of users in
the system. However, this scheme has high complexity in
encoding/decoding and is difficult to be implemented. Therefore,
a suboptimal and low-complexity zero-forcing (ZF) beamforming
technique was proposed in \cite{TYoo2006} which also achieves the optimal growth
rate of the sum rate. The results both in DPC and ZF schemes were based on the full
channel state information (CSI) assumption at the transmitter and thus
the users are required to feedback perfect CSI to the BS. Although
the optimal sum rate can be achieved, the feedback load will
increase linearly with the number of users.

Various approaches have been proposed to limit the amount of
feedback load and investigate sum rate loss. In
\cite{Jindal06}\cite{Goldsmith07}, the quantized channel direction
information (CDI) was used to characterize the sum rate loss when the
ZF technique is considered. It was shown that the number of feedback bits of each
user needs to be increased linearly with the transmission power to maintain a constant sum rate loss. Another low-feedback-rate and practical scheme,
orthogonal random beamforming (ORB) was proposed in \cite{Hassibi05}.
In ORB, each user only feeds back the CSI and the beam index of its favorite beam to the BS. Therefore, the total amount of feedback
can be reduced from $MK$ CSI values to $K$ CSI values. Besides, the
sum rate loss is negligible when the number of users is
large \cite{M_Pugh10}. In an effort to further reduce feedback load,
a threshold based mechanism was proposed in \cite{Gesbert04} such that a user does not feed back when its CSI is below the threshold. In
that work, the design of
threshold does not take the scheduling algorithm into account. In
\cite{JH2010}, multiple thresholds were proposed. The design of multiple thresholds was based on the
order statistics of the signal-to-interference-noise-ratio (SINR) assuming that
greedy sum rate scheduling was performed in a {\it homogeneous} network where the users' channel gain distributions are the same.
Exhaustive search was used to obtain the thresholds, which
resulted in high computational complexity.

In this work, the closed form solution of the multiple thresholds in
\cite{JH2010} is found, and a more realistic assumption on the
channel distributions of the users is considered. We assume that the
distributions of the users' signal-to-noise rations (SNRs) are independent and
non-identical. Since the computational complexity is too high to
consider all the users' non-identical channel statistics, the
statistics of the users' channels are divided into multiple clusters.
The statistics in each cluster are used to calculate the corresponding threshold.
Each cluster corresponds to an SNR range and is quantized with a few
bits for differentiating the users' SNR levels falling in the same
cluster. In addition, the sum rate loss due to setting thresholds is
investigated. The performance of the proposed scheme is
compared with that of the conventional feedback scheme and single
threshold feedback scheme \cite{Gesbert04} in terms of sum rate,
feedback load and efficiency. Through simulations, it is shown that
the proposed cluster-based feedback scheme is more efficient than
conventional feedback schemes and achieves higher sum rate than the
single threshold feedback scheme.
\section{System Model}
The multiuser multiple-input multiple-out (MIMO) downlink system
is considered. The BS is equipped with $M$ antennas and there are
$K$ users, each having $N$ receive antennas. According to the ORB
strategy for multiuser transmission, the BS uses a precoding matrix
${\mathbf {W}}=[{\mathbf{w}}_1,{\mathbf{w}}_2,\ldots,{\mathbf{w}}_M]$
to simultaneously transmit signals, where ${\mathbf {w}}_{i}
\in{{\mathbb{C}}^{M\times1}}, i=1,2,\ldots,M$, are random orthogonal
vectors generated from isotropic distribution \cite{Marzetta1999}.
Let ${\mathbf{s}}=\sum_{m}{\mathbf{w}}_{m}s_{m}$ be the $M\times1$
vector of the transmitted signal, where $s_{m}$ is the $m$th transmitted
symbol.
The received signal for user $k$ is given  by
\begin{eqnarray}
{\mathbf{y}}_k&=&{\mathbf{H}}_{k}{\mathbf{Ws}}+{\mathbf{n}}_k
\end{eqnarray}
where ${\mathbf{H}}_{k}$ denotes the channel matrix between the
BS and user $k$. The elements of the channel matrix ${\mathbf{H}}_{k}$
are assumed to be independent identically distributed (i.i.d.)
complex Gaussian random variables with zero mean and variance $\sigma_{k}^2$. The noise term
for user $k$, denoted ${\mathbf{n}}_{k}$, is modeled as
i.i.d. zero mean complex Gaussian with covariance
matrix
${\mathbb{E}}[{\mathbf{n}}_{k}{\mathbf{n}}_{k}^{H}]=\sigma_{N}^2{\mathbf{I}}, \forall k$, where ${\mathbb{E}}$ denotes the expectation operation and $(\cdot)^{H}$
represents the transpose conjugate.

User $k$ uses the ZF receiver \cite{Heath01} to perform channel
inversion to the received signal ${\mathbf{y}}_{k}$. Thus, the
received signal after ZF receiver is given by
\begin{eqnarray}
{({\mathbf{H}}_{k}{\mathbf{W}})}^{\dagger}{\mathbf{y}}_k={\mathbf{s}}+{({\mathbf{H}}_{k}{\mathbf{W}})}^{\dagger}{\mathbf{n}}_k
\end{eqnarray}
where
${\mathbf{H}}^{\dagger}=({\mathbf{H}}^{H}{\mathbf{H}})^{-1}{\mathbf{H}}^{H}$
is the pseudo-inverse of ${\mathbf{H}}$. Under the equal power
assumption, let the transmit energy of each antenna be
$\frac{P}{M}$, where $P$ is the total transmit power at the BS.
The SNR for the $k$th user at the $m$th spatial channel is denoted
by $X_{m,k}$
\begin{eqnarray}
X_{m,k}=\frac{P/M}{\sigma_{N}^2[(({\mathbf{H}}_{k}{\mathbf{W}})^{H}({\mathbf{H}}_{k}{\mathbf{W}}))^{-1}]_{m}}
& m=1,2,\ldots,M
\end{eqnarray}
where $[{\mathbf{A}}]_{m}$ denotes the $m$th diagonal element of
matrix ${\mathbf{A}}$. Assuming $N\geq M$, it is well known
that $X_{m,k}$ is a chi-square random variable with $2(N-M+1)$
degrees of freedom \cite{R_W_Heath02}\cite{TT2007}. Then the
probability density function (PDF) of $X_{m,k}, \forall m$, can be expressed as
\begin{equation}
f_{k}{(x)}=\frac{\lambda_{k}^{N-M+1}x^{N-M}e^{-\lambda_{k}x}}{(N-M)!}
 \label{chi_square_distribution}\end{equation}
where $\lambda_{k}=\frac{M\sigma_{N}^2}{P\sigma_{k}^2}$. For
simplicity, we drop the index $m$, denote $X_{m,k}$ as $X_{k}$, and
restrict our analysis for the case $N=M$. Therefore, the distribution
of $X_{k}$ becomes the exponential distribution with
parameter $\lambda_{k}$.

Let ${X_{(1)},X_{(2)},\cdots, X_{(K)}}$, be the order statistics of
independent continuous random variables $X_{1},X_{2},\cdots,X_{K}$,
with the PDF ({\ref{chi_square_distribution}) in decreasing order,
i.e., ${X_{(1)}\geq X_{(2)}\geq \cdots \geq X_{(K)}}$. The
cumulative distribution function (CDF) of the largest order statistics
$X_{(1)}$ can be shown as
\begin{eqnarray}
F_{X_{(1)}}(x)=\prod_{i=1}^{K}(1-e^{-\lambda_{i}x})
\end{eqnarray}
and its corresponding PDF can be expressed as
\begin{eqnarray}
f_{X_{(1)}}(x)=\frac{1}{(K-1)!}\sum_{{\mathbf{T}}}F_{i_{1}}(x)\cdots
F_{i_{K-1}}(x)f_{i_{K}}(x)
\end{eqnarray}
where $\sum_{\mathbf{T}}$ denotes the summation over all $K!$
permutations $(i_{1},i_{2},\ldots,i_{K})$ of $(1,2,\ldots,K)$.
Applying the maximum sum rate scheduling algorithm, the sum rate of
the system can be written as follows:
\begin{eqnarray}
R&=&{\mathbb{E}}\left\{\sum_{m=1}^{M}\log_{2}\left(1+\max_{1\leq k\leq K}{X_{m,k}}\right)\right\}        \nonumber \\
&=& {M}\int_{0}^{\infty}\log_{2}(1+x)f_{X_{(1)}}(x)dx.
\label{max_sum_rate_formula}
\end{eqnarray}
In order to achieve the sum rate described in
(\ref{max_sum_rate_formula}), each user should feed back the index of the precoding vector on which it sees the highest SNR,
and the corresponding SNR value.
\begin{figure}[!t]
\centering
\includegraphics[width=0.4\textwidth]{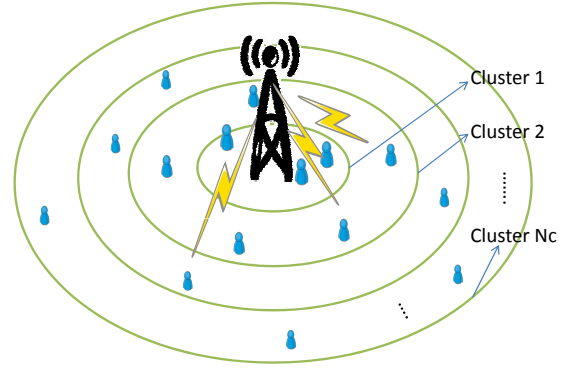}
\caption{Cluster-based Model.} \label{fig:
Cluster_based_feedback_model}
\end{figure}

\section{Cluster-based Feedback Model}
The transmission and feedback procedure can be described as follows.
From the previous channel condition or location information feedbacks from the users, the BS derives the users' mean SNR. The BS then groups the users' mean SNRs into $N_c$ clusters according to their magnitudes.
The mean SNRs in each cluster are similar in quantity and are used to derive one SNR threshold, denoted $r_{c,i}$ for cluster $i$. Note that derivation of the users' mean SNRs and the cluster thresholds is done periodically. Derivation of the users' mean SNRs does not need to be very accurate. The BS broadcasts periodically the threshold
set $\{r_{c,1}\geq r_{c,2}\geq\cdots\geq r_{c,N_{c}}\}$ to the users.

At every feedback instant, each user compares its measured instantaneous SNR with the
thresholds to determine which cluster its instantaneous SNR belongs
to, and feeds back to the BS the cluster index and the quantization bits of the instantaneous SNR
for this cluster. When a user's SNR is smaller
than $r_{c,N_{c}}$, the user does not feedback to the BS. The proposed
procedure takes a little downlink bandwidth for the BS to
periodically broadcast the threshold set, in exchange of reduction
of the uplink feedback bandwidth.

To derive the thresholds, the overall statistics of users are
divided into multiple clusters according to the mean SNRs
$\frac{1}{\lambda_{1}}, \frac{1}{\lambda_{2}}, \cdots,
\frac{1}{\lambda_{K}}$. The means of the random variables ${X_{1},
X_{2},\cdots,X_{K}}$ are ranked with decreasing order as
$\frac{1}{\lambda_{(1)}}\geq \frac{1}{\lambda_{(2)}}\geq \cdots\geq
\frac{1}{\lambda_{(K)}}$ and uniformly divided into $N_{c}$ clusters
with size $L=K/N_{c}$ as
$\left\{\frac{1}{\lambda_{(1)}},\frac{1}{\lambda_{(2)}},\cdots,\frac{1}{\lambda_{(L)}}\right\}$,$\cdots$,
$\left\{\frac{1}{\lambda_{(L(N_{c}-1)+1)}},\frac{1}{\lambda_{(L(N_{c}-1)+1)}},\cdots,\frac{1}{\lambda_{(K)}}\right\}$.
To simplify the notation, we let random variable $Y_{m}^{n}$
represent the $m$th statistic in the $n$th cluster. Then,
$Y_{m}^{n}$ is exponentially distributed with mean value
$\frac{1}{\lambda_{(L(n-1)+m)}}$. The random variables in cluster
$n$ are $\left\{Y_{1}^{n},Y_{2}^{n},\cdots,Y_{L}^{n}\right\}$, and
let $Y_{(1)}^{n}\geq Y_{(2)}^{n}\geq\cdots\geq Y_{(L)}^{n}$ be the
order statistics of them. For a measured instantaneous SNR $y_i$ at
user $i$, we say that the rank of $y_i$ in cluster $n$ is $d$ if
$\underbrace{Y_{(1)}^{n}\geq Y_{(2)}^{n} \geq Y_{(3)}^{n}\cdots
}_{\text{$(d-1)$ variables}}\geq y_{i} \geq
\underbrace{Y_{(d+1)}^{n}\cdots \geq Y_{(K)}^{n}}_{\text{ $(L-d)$
variables}}$. In the following, we will discuss how to derive the
threshold in each cluster for heterogenous and homogeneous channel
distributions to reduce the feedback load.
\subsection{Heterogenous Case}
\subsubsection{Cluster-based Type-I}
With maximum sum rate scheduling, a user will be scheduled when its SNR on a particular precoding vector is the highest among the users. On the other hand, if a user has a low SNR, it is unlikely to be scheduled. For this user, feeding back CSI is wasteful of the uplink radio resource. The threshold of each cluster is designed according to the probability of a user's measured instantaneous SNR being a particular rank in that cluster.
Let $P_{m,n}^{i}(r)$ denote the probability that user $m$ is ranked $n$ among the $L$ users in
cluster $i$ when its instantaneous SNR is $r$.
\begin{eqnarray}
\begin{array}{l}
P_{m,n}^{i} (r)=P\{Y_{m}^{i}=Y_{(n)}^{i}|Y_{m}^{i}=r\} \nonumber  \\
=\displaystyle{\sum_{\mathbf{S}}}P\{\underbrace{Y_{t_{1}}^{i}\geq
Y_{t_{2}}^{i}\geq \cdots \geq Y_{t_{n-1}}^{i}}_{\text{$(n-1)$
variables}}\geq r
 \geq \underbrace{Y_{t_{n+1}}^{i} \cdots \geq Y_{t_{L}}^{i}}_{\text{ $(L-n)$ variables
 }}\}\\
=\displaystyle{\frac{{\displaystyle{\sum_{\mathbf{S}}}e^{{-\displaystyle{\sum_{a=1}^{n-1}\lambda_{(L(i-1)+t_{a})}r}}}\displaystyle{\prod_{b=n+1}^{L}(1-e^{{-\lambda_{(L(i-1)+t_{b})}r}})}}}{(n-1)!(N_{c}-n)!}}\\
\end{array}
\end{eqnarray}
where $\sum_{\mathbf{S}}$ denotes the summation over all
permutations $(t_{1},\cdots,t_{n-1},t_{n+1},\cdots,t_{L})$ of
$(1,2,\cdots,m-1,m+1,\cdots,L)$ in the cluster $i$. For
example, when the instantaneous SNR of user $t$ in cluster $1$ is infinity,
the rank of user $t$ among the $L$ users in cluster $1$ is one
with probability one, i.e.,
$P_{t,1}^1 (\infty) = P\{Y_{t}^{1}=Y_{(1)}^{1}|Y_{t}^{1}=\infty\}=1$.
Being valid conditional probabilities, the $P_{m,n}^{i}(r)$'s satisfy
\begin{eqnarray}
\sum_{n=1}^{L}P_{m,n}^{i} (r)=1 ,&m=1,2,\cdots,L.
\end{eqnarray}
The most probable rank of user $m$ in cluster $i$ when its instantaneous SNR is $r$ can be obtained by
\begin{equation}
\hat{n}=\arg \max_{n\in \{1,2,\ldots,L \}} P_{m,n}^{i} (r).
\end{equation}

Let $Q_{m,n}^{i}$ be defined such that $P_{m,n}^{i} \left( Q_{m,n}^{i}\right)=P_{m,n+1}^{i} \left( Q_{m,n}^{i}\right)$. It can be seen that when the instantaneous SNR of user $m$ falls in the range $\left[Q_{m,n}^{i},Q_{m,n-1}^{i}\right)$, the most probable rank of user $m$ in cluster $i$ is $n$.
For maximum sum rate scheduling, we let only the users who are most likely to be rank one in each cluster to feedback.
Therefore, a threshold for each cluster $i$ is set as
\begin{eqnarray}
r_{c,i}=\max\{Q_{1,1}^{i},Q_{2,1}^{i},\cdots,Q_{L,1}^{i} \} &
i=1,2,\cdots,N_{c}
\end{eqnarray}
which satisfies the following inequality
\begin{eqnarray}
r_{c,1}\geq r_{c,2}\geq\cdots\geq r_{c,N_{c}}.
\end{eqnarray}
To reduce the computational complexity, another
simple and low complexity method is proposed in the next section.
\subsubsection{Cluster-based Type-II} \label{type-II}
We approximate by assuming that the random
variables $Y_{m}^{i}, m=1,2,\cdots,L$ in cluster $i$ have the same
exponential distribution with mean $\mu_{c_{i}}$ obtained by
\begin{eqnarray}
\mu_{c_{i}}=\frac{\sum_{m=1}^{L}\lambda_{(L(i-1)+m)}}{L}.
\end{eqnarray}
Then, $P_{m,n}^{i} (r)$ is the same for all users in a cluster, and can be obtained by
\begin{eqnarray}
P_{m,n}^{i} (r)=\frac{(L-1)!\exp{(-\mu_{c_{i}}r)}^{n-1}{(1-\exp(-\mu_{c_{i}}r))}^{L-n}}{(n-1)!(L-n)!}.
\end{eqnarray}
With this approximation, the closed-form solution of $r_{c,i}$ can be derived as
\begin{eqnarray}
r_{c,i}=Q_{1,1}^{i}=\frac{1}{\mu_{c_{i}}}\ln{L} , i=1,2,\cdots,
N_{c}. \label{closed_form_threshold}
\end{eqnarray}
\subsection{Homogeneous Case}
The homogeneous case can be viewed as a special case of the
cluster-based type-II scenario using only one cluster $(N_{c}=1,
K=L)$. The mean values of the random variables $Y_{m}^{1},
m=1,2,\cdots,K$ are the same, i.e.,
$1/\lambda=1/\lambda_{1}=\cdots=1/\lambda_{K}$. As in \cite{JH2010},
multiple thresholds are set according to the most probable rank as
\begin{eqnarray}
r_{c,p}=Q_{1,p}^{1}=\frac{1}{\lambda}\ln{\left(\frac{K}{p}\right)},  & p=1,2, \cdots,
N_{c}.
\end{eqnarray}

For all cases, at each feedback instant, when a user's instantaneous
SNR is greater than the smallest threshold $r_{c,N_{c}}$, the user needs to
feed back to the BS using $B_{C}=\log_{2}\lceil N_{c} \rceil $ bits
to indicate which region (between two adjacent thresholds) its instantaneous SNR belongs to. In addition, in order
to differentiate the users SNR in the same region $i$, the
region $i$ which is represented by $\mathfrak{C_{i}}$ is further
quantized with $b_{i}$ bits. Therefore, the feedback bits of each
user include two parts: one is the region index with $B_{C}$ bits
and the other is the quantization bits for that region.

\section{Analysis of Sum Rate Loss }
A natural question to ask
is how many clusters (thresholds) should be set.
Obviously, if many clusters are used, the number of region index bits $B_{C}$
is increased. On the other hand, when a small number of
clusters are used, the smallest threshold $r_{c,N_{c}}$ becomes
large. Then, the sum rate loss caused by the threshold $r_{c,N_{c}}$ may
increase. Therefore, an appropriate number of clusters is important
for the design. We now analyze the sum rate loss of the
system caused by the smallest threshold $r_{c,N_{c}}$. We define a
random variable $Z_{i}$ as follows:
\begin{eqnarray}
Z_{i}=\left\{\begin{array}{cc}
    0 ,& X_{i}>r_{c,N_{c}} \\
    X_{i} ,& X_{i}\leq r_{c,N_{c}} \\
  \end{array}\right..
\end{eqnarray}
According to the maximum sum rate criterion, the exact SNR loss
for the system is $Z_{(1)}=\max\{Z_{1},Z_{2},\cdots,Z_{K}\}$. The
probability of the rate loss event can be obtained by
\begin{eqnarray}
P_{L}=P\{X_{(1)}\in
(0,r_{c,N_{c}})\}=\displaystyle{\prod_{i=1}^{K}}(1-\exp{(-\lambda_{i}r_{c,N_{c}})}).
\end{eqnarray}
Therefore, the sum rate loss can be expressed as
\begin{eqnarray}
\Delta R=ME\{\log_{2}(1+Z_{(1)})\}P_{L}.
\end{eqnarray}

\begin{theorem}
\cite{D_B2006} Let the means and variances of the
random variables $Z_1, Z_2, \ldots, Z_K$ be
${\mathbf{\mu}}=(\mu_{1},\mu_{2},\cdots,\mu_{K})$ and
${\mathbf{\sigma^{2}}}=(\sigma_{1}^2,\sigma_{2}^2,\cdots,\sigma_{K}^2)$, respectively.
The closed form upper bounds on the expected value of the largest
order statistic is:
\begin{eqnarray}
E\{Z_{(1)}\} &\leq& {\sum_{i=1}^{K}\left\{\begin{array}{c}
                 \displaystyle{\frac{\mu_{i}+\sqrt{(\mu_{i}-T)^{2}+\sigma_{i}^{2}}}{2}}\nonumber \\
              \end{array}\right\}+\frac{(2-K)T}{2}} \nonumber \\ &\triangleq&
              \mu_{Z_{(1)}^{U}}
 \label{expacation_upper_bound}
\end{eqnarray}
where $T=\max_{1\leq j\leq
K}\{\mu_{j}+\frac{K-2}{2\sqrt{K-1}}\sigma_{j}\}$.
\end{theorem}
Applying Jensen's
inequality and the upper bound of $E\{Z_{(1)}\}$, the sum rate loss on a certain beam can be
bounded by
\begin{eqnarray}
\frac{\Delta R}{M}&=&E\{\log_{2}(1+Z_{(1)})\}P_{L}\leq
\log_{2}(1+E\{Z_{(1)}\})P_{L} \nonumber \\
&\leq&\log_{2}(1+\mu_{Z_{(1)}^{U}})P_{L}.
\label{rate_loss_upper_bound}
\end{eqnarray}
Using (\ref{rate_loss_upper_bound}), for a given tolerable sum rate loss $\Delta R_{U}$, the minimum
number of clusters (thresholds) required can be determined.

\begin{proposition}
When the total number of users approaches infinity,
the sum rate loss caused by the smallest finite threshold $r_{c,N_{c}}$ is negligible. Thus, the full multiuser diversity can be achieved.
\end{proposition}
\begin{proof}
When the total number of users $K$ approaches to infinity,
$\displaystyle{\lim_{K\rightarrow \infty}}P_{L}=0$. In
addition, $Z_{(1)}$ is bounded by $r_{c,N_{c}}$. The sum rate
loss on a certain beam becomes zero
\begin{eqnarray}
\lim_{K\rightarrow \infty}\frac{\Delta R}{M}&=&\lim_{K\rightarrow \infty}E\{\log_{2}(1+Z_{(1)})\}P_{L} \nonumber \\
&\leq&\lim_{K\rightarrow \infty}
\log_{2}(1+E\{Z_{(1)}\})P_{L} \nonumber \\
&\leq&\lim_{K\rightarrow \infty}\log_{2}(1+r_{c,N_{c}})P_{L} \nonumber \\
&=&0.
\end{eqnarray}
%
\end{proof}

\section{Bit Allocation and Feedback Load Analysis }
\subsection{Bit Allocation}
Let $r_{c,0} = \infty$. Assume that the SNR region $i$ $[r_{c,i}, r_{c,i-1})$, denoted
$\mathfrak{C}_{i}, i=1,2,\cdots,N_{c}$, is quantized with $b_{i}$ bits. In region $\mathfrak{C}_{i}$, the quantization
levels using $b_{i}$ bits are expressed by $q_{t}^{i}, t=1,2,\cdots,2^{b_{i}}$, obtained by a pdf quantizer \cite{Jose06}. Thus, each level will occur with the same probability.
Under the per user average feedback load constraint $C_{k}$ for user $k$, the available bits will be assigned to the regions to maximize sum rate. Let
$P_{\mathfrak{C}_{i}}=P\{X_{(1)}\in \mathfrak{C}_{i} \},
i=1,2,\cdots,N_{c}$. The bit allocation problem can be
described as follows:
\begin{eqnarray}
&\displaystyle{\max_{(b_{1},b_{2},\cdots,b_{N_{c}})}}& {M\sum_{i=1}^{N_{c}}\frac{P_{\mathfrak{C}_{i}}}{2^{b_i}}\sum_{t=1}^{2^{b_{i}}}\log(1+q_{t}^{i})} \nonumber \\
& \text{s.t}& \sum_{i=1}^{N_{c}}P_{\mathfrak{C}_{k,i}}b_{i} \leq
C_{k}, k=1,2,\cdots,K,   \label{opt_bit}
\end{eqnarray}
where $P_{\mathfrak{C}_{k,j}}=P\{X_{k} \in
\mathfrak{C}_{j}\}=e^{-\lambda_{k}r_{c,j}}-e^{-\lambda_{k}r_{c,j-1}}$.
The closed-form integer solution for (\ref{opt_bit}) does not exist.
We use exhaustive search to find the optimal bit allocation set to
maximize the sum rate.
\subsection{Feedback Load Analysis}
When a user's instantaneous SNR is smaller than the smallest threshold
$r_{c,N_{c}}$, the user does not need to feedback.
Thus, the feedback probability for user $k$ is
$e^{-\lambda_{k}r_{c,N_{c}}}$. The average total feedback
load can be expressed as
\begin{eqnarray}
F_{b}&=&M\sum_{k=1}^{K}\{e^{(-\lambda_{k}r_{c,N_{c}})}{(B_{C}+C_{k})}\}
\end{eqnarray}
\section{Simulation Results}
In this section, we show the sum rate and feedback load performance
for the proposed feedback scheme. The BS is equipped with $M=4$
antennas and the total number of users is $K=10\sim100$. In order to
perform ZF beamforming, we let the number of receive antennas $N=4$.
The total transmit power $P$ is 10W, while the additive white Gaussian noise power at the receivers $\sigma_N^2$ is 1W. Note that these numbers are selected only for illustration purpose.
The elements of the channel matrix
$\mathbf{H}_{i}$ for the $i$th user are assumed to be i.i.d. complex
Gaussian distribution with zero and variance $\sigma_{i}^2$, where
$\sigma_{i}^2$ are drawn uniformly from the interval $[0,1]$ to
model heterogeneous channels. The number of clusters in type-I and
type-II schemes is set to four (thus $B_{C}=2$) according to the
tolerable sum rate loss $\Delta R_{U}=10^{-2}$ $bps/Hz$. The
feedback load constraint of user $k$ is $C_{k}=0.8, k=1,2,\cdots,K$.
In the simulation, the proposed feedback schemes are compared with
the conventional scheme and the single-threshold schemes proposed in
\cite{Gesbert04}. In the conventional feedback scheme, no matter
what the instantaneous SNR is, it is always quantized with $3$ bits.
In the single-threshold scheme, the region $(r_{th},\infty)$ of the
SNR is quantized using $3$ bits. The single threshold $r_{th}$ is
established according to the scheduling outage probability
$P_{out}$.

In Fig.~\ref{fig: sum_rate}, the type-I feedback scheme achieves the
highest sum rate, and the low-complexity type-II scheme
achieves almost the same rate as the type-I scheme. As shown in Fig.~\ref{fig: feedback_load}, the feedback
load of the conventional scheme increases linearly with the number of users. Using our
proposed schemes, the total number of feedback bits can be
dramatically reduced. Overall, the proposed schemes use fewer bits to achieve higher sum rate than the conventional scheme. The single-threshold scheme has lower feedback load, but suffers significantly in the sum rate performance.

In Fig.~\ref{fig: efficiency}, we plot the sum rate vs. the total
feedback load as an indication of the efficiency. The type-I scheme
is the most efficient (i.e., making best use of the feedback bits), but has high computational complexity.
The low-complexity type-II scheme not only achieves high sum rate
but also reduces the feedback load significantly. The single-threshold
feedback scheme only achieves a sum rate of about $11(bps/Hz)$ with a small feedback load.
\begin{figure}
\centering
\includegraphics[width=0.42\textwidth]{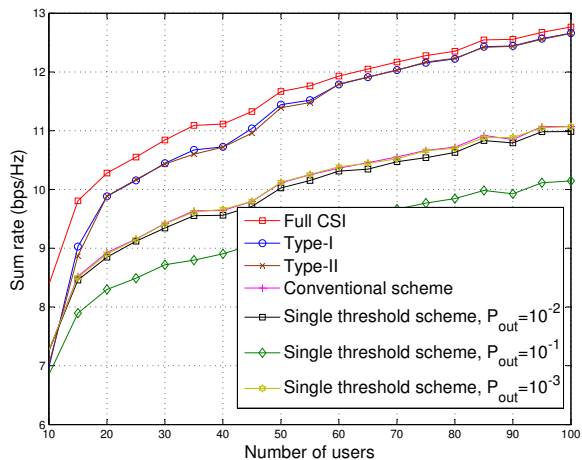}
\caption{Sum rate comparison between different feedback schemes.}\label{fig: sum_rate}
\end{figure}

\begin{figure}
\centering
\includegraphics[width=0.42\textwidth]{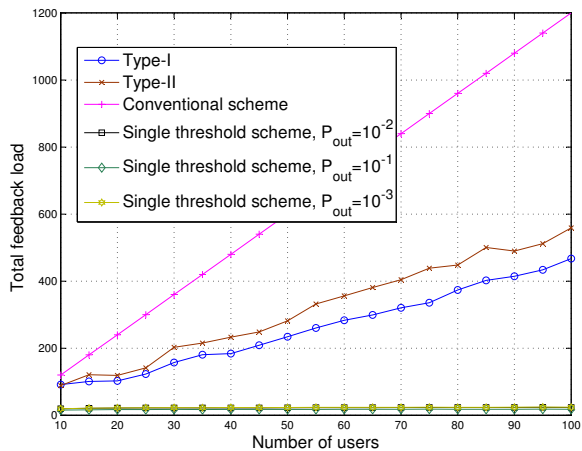}
\caption{Feedback load comparison between different feedback
schemes.}\label{fig: feedback_load}
\end{figure}

\begin{figure}
\centering
\includegraphics[width=0.42\textwidth]{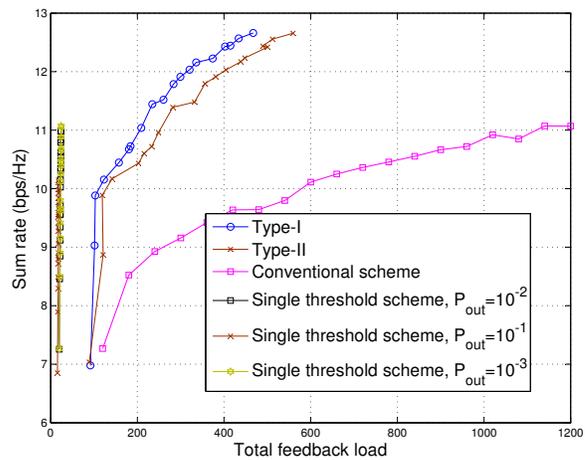}
\caption{The efficiencies of different feedback schemes.}\label{fig: efficiency}
\end{figure}

\section{Conclusion}
In this paper, we investigated the feedback load reduction problem in
multiuser MIMO broadcast system. We
proposed a cluster-based feedback scheme to reduce the feedback load
in heterogenous and homogeneous channels. The bit allocation problem
for the multiple-cluster feedback scheme was also discussed. The
simulation results showed that, compared to the existing feedback schemes, the cluster-based feedback scheme can make the best use of the feedback bits to achieve good feedback load reduction while maintaining good sum rate
performance.
\bibliographystyle{IEEEbib}
\bibliography{ref}

\begin{thebibliography}{10}

\bibitem{Pv_Tse02}
P.~Viswanath, D.~Tse, and R.~Laroia,
\newblock ``{Opportunistic beamforming using dumb antennas},''
\newblock {\em IEEE Trans. on Inf. Theory}, vol. 48, no. 6, pp. 1277--1294,
  2002.

\bibitem{M_H_M_Costa1983}
M.~H.~M. Costa,
\newblock ``{Writing on dirty paper},''
\newblock {\em IEEE Trans. on Inf. Theory}, vol. 29, no. 3, pp. 439--441, May
  1983.

\bibitem{J_Cioffi_04}
W.~Yu and J.~Cioffi,
\newblock ``{Sum capacity of Gaussian vector broadcast channels},''
\newblock {\em IEEE Trans. on Inf. Theory}, vol. 50, pp. 1875--1892, 2004.

\bibitem{TYoo2006}
T.~Yoo and A.~Goldsmith,
\newblock ``{On the optimality of multiantenna broadcast scheduling using
  zero-forcing beamforming},''
\newblock {\em IEEE J. Select. Areas Commun}, vol. 24, pp. 528--541, Mar. 2006.

\bibitem{Jindal06}
N.~Jindal,
\newblock ``{MIMO Broadcast Channels With Finite-Rate Feedback},''
\newblock {\em IEEE J. Select. Areas Commun}, vol. 52, no. 11, pp. 5045--5060,
  2006.

\bibitem{Goldsmith07}
T.~Yoo, N.~Jindal, and A.~Goldsimth,
\newblock ``{Multi-Antenna Downlink Channels with Limited Feedback and User
  Selection},''
\newblock vol. 25, no. 7, pp. 1478--1491, Sep 2007.

\bibitem{Hassibi05}
M.~Sharif and B.~Hassibi,
\newblock ``{On the capacity of MIMO broadcast channel with partial side
  information},''
\newblock {\em IEEE Trans. on Inf. Theory}, vol. 51, no. 2, pp. 506--522, Feb.
  2005.

\bibitem{M_Pugh10}
M.~Pugh and B.~D. Rao,
\newblock ``{Reduced Feedback Schemes Using Random Beamforming in MIMO
  Broadcast Channels},''
\newblock {\em IEEE Trans. on Signal Processing}, vol. 58, no. 3, pp.
  1821--1832, March 2010.

\bibitem{Gesbert04}
D.~Gesbert and M.~S. Alouini,
\newblock ``{How much feedback is multi-user diversity really worth?},''
\newblock {\em IEEE Int. Conf. on Commun.(ICC)}, vol. 1, pp. 234--238, June
  2004.

\bibitem{JH2010}
J.-H. Li and H.-J. Su,
\newblock ``{A Multi-threshold Scheme for Feedback Load Reduction in Multiuser
  MIMO Broadcast Channel},''
\newblock {\em IEEE PIMRC}, Sept. 2010.

\bibitem{Marzetta1999}
T.~L. Marzetta and B.~M. Hochwald,
\newblock ``{Capacity of a mobile multiple-antenna communication link in
  Rayleigh flat fading},''
\newblock {\em IEEE Trans. on Inf. Theory}, vol. 45, pp. 139--157, 1999.

\bibitem{Heath01}
R.~W. Heath, M.~Airy, and A.~J. Paulraj,
\newblock ``{Multiuser diversity for MIMO wireless systems with linear
  receiver},''
\newblock {\em in proc. of the Asil. Conf. on Sig. Sys. and Comp.}, vol. 2, pp.
  1194--1199, Nov. 2001.

\bibitem{R_W_Heath02}
D.~A. Gore, R.~W. Heath, and A.~J. Paulraj,
\newblock ``{Transmit Selection in Spatial Multiplexing Systems},''
\newblock {\em IEEE Communication Letter}, vol. 6, pp. 491--493, Nov. 2002.

\bibitem{TT2007}
T.~Tang, R.~W. Heath, S.~Cho, and S.~Yun,
\newblock ``{Opportunistic Feedback for Multiuser MIMO Systems With Linear
  Receivers},''
\newblock {\em IEEE Trans. on Communications}, vol. 55, no. 5, May. 2007.

\bibitem{D_B2006}
K.~Natarajan D.Bertsimas and C.-P.Teo,
\newblock ``{Tight bounds on expected order statistics},''
\newblock {\em Probab.Eng.Inform.Sc.}, vol. 20, no. 4, pp. 667--686, 2006.

\bibitem{Jose06}
J.~L. Vicario, R.~Bososio, and C.~Anton-Haro,
\newblock ``{A throughout analysis for opportunistic beamforming with quantized
  feedback},''
\newblock {\em IEEE PIMRC}, pp. 1--5, Sept. 2006.

\end{thebibliography}
\end{document}